\documentclass[11pt,letterpaper]{article}
\usepackage{graphicx,amssymb,amsmath,url, amsthm}
\usepackage[dvipsnames]{xcolor}
\usepackage{fullpage}
\usepackage[font=small]{caption}
\usepackage{lineno}
\usepackage{hyperref}
\usepackage[vlined]{algorithm2e}
\usepackage[inline]{enumitem}

\SetKw{Continue}{continue}
\newtheorem{definition}{Definition}[section]

\newtheorem{lemma}{Lemma}
\newtheorem{theorem}{Theorem}

\title{Infinite All-Layers Simple Foldability}

\date{}
\author{
  Hugo A. Akitaya\thanks{Tufts University, \protect\url{hugo.alves_akitaya@tufts.edu}}
\and
  Cordelia Avery\thanks{CSAIL, Massachusetts Institute of Technology, \protect\url{{cavery,jbergero,edemaine,jkopin,jasonku}@mit.edu}}
\and
  Joseph Bergeron\footnotemark[2]
\and
  Erik D. Demaine\footnotemark[2]
\and
  Justin Kopinsky\footnotemark[2]
\and
  Jason Ku\footnotemark[2]
}

\index{Akitaya, Hugo}
\index{Avery, Cordelia}
\index{Bergeron, Joseph}
\index{Demaine, Erik D.}
\index{Kopinsky, Justin}

\begin{document}
\thispagestyle{empty}
\maketitle

\begin{abstract}
We study the problem of deciding whether a crease pattern can be folded by simple folds (folding along one line at a time) under the \emph{infinite all-layers} model introduced by~\cite{SimpleFolds_JCDCGGG2016}, in which each simple fold is defined by an infinite line and must fold all layers of paper that intersect this line.
This model is motivated by folding in manufacturing such as sheet-metal bending.
We improve on~\cite{arkin04} by giving a deterministic $O(n)$-time algorithm to decide simple foldability of 1D crease patterns in the all-layers model.
Then we extend this 1D result to 2D, showing that simple foldability in this model can be decided in linear time for unassigned axis-aligned orthogonal crease patterns on axis-aligned 2D orthogonal paper.
On the other hand, we show that simple foldability is strongly NP-complete if a subset of the creases have a mountain--valley assignment, even for an axis-aligned rectangle of paper.
\end{abstract}

\section{Introduction}

A classic problem in computational origami (indeed, the topic of the very first paper in the field \cite{bern1996})
is \emph{flat foldability}: given a crease pattern (planar straight-line graph) on a polygonal piece of paper $P$, can $P$ be folded flat isometrically without self-intersection while creasing at all creases (edges) in the crease pattern?
The decision problem is NP-hard~\cite{bern1996}, even if the paper $P$ is an axis-aligned rectangle and the creases are at multiples of $45^\circ$ \cite{Akitaya2016}.
But even when a crease pattern does fold flat, the motion to achieve that folding can be complicated \cite{PaperReachability_CCCG2001}, making the process impractical in some physical settings.

Motivated by practical folding processes in manufacturing such as sheet-metal bending,
Arkin et al.~\cite{arkin04} introduced the idea of \emph{simple foldability}---flat foldability by a sequence of simple folds.
Informally, a \emph{simple fold} is defined by a line segment and rotates a portion of the paper around this line by $\pm 180^\circ$, while avoiding self-intersection; refer to Section~\ref{sec:def} for formal definitions.
The problem also makes sense with 1D paper,
where $P$ is a line segment and creases are defined by points in~$P$.
Arkin et al.\ defined several models of simple folds and, for many models, showed that simple foldability is polynomial for 1D paper, polynomial for rectangular paper with axis-aligned creases, weakly NP-complete for rectangular paper with creases at multiples of $45^\circ$, and weakly NP-complete for orthogonal paper with axis-aligned creases.
In particular, they developed an algorithm to determine simple foldability of a 1D mountain--valley pattern in $O(n\log n)$ deterministic time and $O(n)$ randomized time in the \emph{all-layers} model, requiring that a simple fold through one crease, also folds through all layers overlapping that crease.

\textbf{Algorithmic results.}
In Section~\ref{sec:1D}, we improve on~\cite{arkin04} by giving a deterministic $O(n)$-time algorithm to decide simple foldability of 1D crease patterns in the all-layers model.
This result removes the logarithmic factor from the best previous deterministic
algorithm, or equivalently, removes the randomization from the best previous
$O(n)$ algorithm.

\medskip

In a recent followup to Arkin et al.~\cite{arkin04},
Akitaya et al.~\cite{SimpleFolds_JCDCGGG2016} extended
the list of simple folding models,
and for many models, showed \emph{strong} NP-hardness for 2D paper.
In particular, they introduced the \emph{infinite all-layers} model of simple folds, studied here, which requires that each simple fold is defined by an infinite line, and that all layers of paper intersecting this line must be folded.
This model is probably the most practical of all simple folding models; for example, Balkcom's robotic folding system \cite{balkcom2008} is restricted to this model.
For an axis-aligned rectangle paper with axis-aligned creases (and for 1D paper), infinite and non-infinite simple fold models are equivalent \cite{SimpleFolds_JCDCGGG2016}.

\textbf{Hardness results.}
In this paper, we study the complexity of one of the few remaining open
problems in this area \cite{SimpleFolds_JCDCGGG2016}: infinite all-layers
simple foldability on axis-aligned orthogonal paper with axis-aligned creases
(henceforth referred to as \emph{orthogonal crease patterns}).
In Section~\ref{sec:Unassigned}, we prove that, when the creases are unassigned (can freely fold mountain or valley), this problem can be solved in polynomial (indeed, linear) time. 
On the other hand, we prove in Section~\ref{sec:PartAssign} that, when the creases are \emph{partially assigned} (some creases must fold mountain, some creases must fold valley, while others can freely fold mountain or valley), the problem becomes strongly NP-complete, even for an axis-aligned rectangle of paper (and thus also for the regular all-layers simple fold model \cite{SimpleFolds_JCDCGGG2016}).
Remaining open problems are summarized in Section~\ref{sec:Open}.


\section{Definitions}
\label{sec:def}
We base our definitions on those from \cite{SimpleFolds_JCDCGGG2016}, specializing to the infinite all-layers model.
Define the \emph{folding plane} $\mathbb{P}$ be a copy of $\mathbb{R}^2$,
and define a \emph{piece of paper} $P$ to be a connected polygon in $\mathbb{P}$ possibly with holes.
A \emph{crease pattern} is defined by $(P,\Sigma)$ where $\Sigma$ is the edge set of a planar straight-line graph embedded in~$P$.
Members of $\Sigma$ are interior disjoint line segments contained in $P$ called \emph{creases}.
A \emph{facet} is a connected open set in $P\setminus \Sigma$.
A crease pattern $(P,\Sigma)$  can have an \emph{assignment} $\alpha:\Sigma\rightarrow \{M,V\}$, where $M=1$ (mountain) and $V=-1$ (valley).
A function $\alpha:\Sigma^-\rightarrow \{M,V\}, \Sigma^-\subset \Sigma$ is called a \emph{partial assignment} of $(P,\Sigma)$.

We define an \emph{infinite all-layers simple fold} to be the operation that takes a crease pattern $(P,\Sigma)$, possibly with a (partial) assignment, and a directed infinite line $\ell$ in $\mathbb{P}$, and returns a simpler crease pattern $(P',\Sigma')$ (``using up'' some creases, i.e., $|\Sigma'|<|\Sigma|$) or reports ``illegal'', as defined below.
Let $P_l$ (resp., $P_r$) be the subset of $P$ in the left (resp., right) open half-plane defined by $\ell$.
Let $\Sigma_l$ (resp., $\Sigma_r$) be the subset of $\Sigma$ contained in $P_l$ (resp., $P_r$).
Reflect $P_r$ and $\Sigma_r$ about $\ell$, obtaining $P_r'$ and $\Sigma_r'$.
If a crease $c\in\Sigma_r$ had an assignment, define $\alpha(c')=-\alpha(c)$ where $c'$ is the reflection of $c$.
Two creases \emph{overlap} if their intersection is a non-degenerate line segment.

\begin{definition}[Legal folds]
\label{def:legal}
An all-layers simple fold through $\ell$ is \emph{legal} if it satisfies the following conditions:
\begin{enumerate}[label=(\arabic*)]
\item \label{cond:only-on-crease} Line $\ell$ does not contain any point of a facet of $P$, i.e., $\ell\cap P\subseteq \Sigma\cup \partial P$, where $\partial P$ represents the boundary of $P$;
\item \label{cond:creases-to-creases} No crease in $\Sigma_r'$ (resp., $\Sigma_l$) contains any point of a facet in $P_l$ (resp., $P_r'$).

If a (partial) assignment is given we additionally require the following conditions:
\item \label{cond:mountains-to-mountains} For every pair of assigned creases $c_1$ and $c_2$ in $\ell$,  $\alpha(c_1)=\alpha(c_2)$.
\item \label{cond:consistent-assign} For every pair of assigned creases $(c_1, c_2)$ that are contained in the transitive closure of the overlap relation,  $\alpha(c_1)=\alpha(c_2)$.
\end{enumerate}
Otherwise we call the fold \emph{illegal}.
\end{definition}

We define the crease pattern $(P',\Sigma')$ obtained from a legal fold as follows.
Let $P'$ be the closure of $P_l\cup P_r'$.
Set $\Sigma'$ initially as $\Sigma_l\cup\Sigma_r'$ and successively  merge all pairs of overlapping creases $c_1$ and $c_2$.
By (2), every interior-intersecting pair of creases overlap.
If either $c_1$ or $c_2$ is assigned, define the assignment of the new crease as $\alpha(c_1)$ or $\alpha(c_2)$.
By (4), all creases can be consistently assigned as above regardless of the order of the merges.

In the one-dimensional case, we define $\mathbb{R}^1$ as the folding plane, $P$ as a line segment and $\Sigma$ as a set of points in $P$.
A fold is defined by a point $\ell$ in $P$, which, by (1), must be also a crease.
(3) is vacuously true in the 1D setting.
All the other definitions follow by the definitions of the two-dimensional problem.

Finally we can define the problem at hand.
The \emph{infinite all-layers simple foldability problem} asks whether, given a crease pattern $(P,\Sigma)$, there is a sequence of crease patterns $S=((P_1,\Sigma_1), \ldots, (P_m,\Sigma_m))$
such that $(P,\Sigma)=(P_1,\Sigma_1)$, $\Sigma_m=\emptyset$, and $(P_i,\Sigma_i)$ is the result of a legal infinite all-layers simple fold of $(P_{i-1},\Sigma_{i-1})$ for $i\in\{2,\ldots,k\}$. 

\paragraph{Difference between definitions.}
Akitaya et al.~\cite{SimpleFolds_JCDCGGG2016} defines a flat folded state as an \emph{isometry} of $P$ --- an isometric embedding of $P$ into $\mathbb{P}$ that preserves connectivity between facets and creases which imposes the non-stretching restrictions on the paper --- together with a non-crossing \emph{layer ordering} --- a binary relation between overlapping facets in the isometry describing the above/below relationship between them.
Their ``non-crossing" definition captures the fact that the paper cannot penetrate itself.
A simple fold is an operation that takes a flat folded state and returns another, modeling the rotation of a portion of the paper through an axis while preserving the non-stretching and non-self-penetrating properties.
$P$ can also be defined as an orientable surface with top and bottom sides for which a crease (which can be assigned \emph{mountain} or, resp., \emph{valley}) is the fold through it which brings together the bottom (resp., top) sides of adjacent facets.
However, because this paper focuses on only the infinite all-layers model, we propose a simplified definition of simple folds, without defining a folded state.
Our definition is equivalent in the sense that it preserves simple foldability of input crease patterns and
the sequence $S$ of crease patterns described above can be easily converted into a sequence of flat folded states by defining the isometry based on the sequence of reflections.   The layer ordering can be recovered from the assignment of the creases on $\ell$ (choosing an arbitrary assignment if all creases are unassigned).
This conversion is possible because, in the infinite all-layers model, if two facets overlap in the isometry, they can be considered as ``glued together'' because no simple folds in this model can separate them (in particular, unfolding is not allowed).

\section{1D Crease Patterns}
\label{sec:1D}

In this section, we consider only assigned and unassigned crease patterns.
We build on ideas from~\cite{arkin04}, representing an instance with $n$ creases as a string $S$ of length $2n+1$, denoting by $S[i], i\in\{1,\ldots,2n+1\}$ the $i$th symbol in $S$. 
For even $i$, $S[i]$ represents the assignment of the $(i/2)$nd crease or $0$ if unassigned.
For odd $i$, $S[i]$ represents the distance between the $\lfloor i/2\rfloor$th and the $(\lfloor i/2\rfloor+1)$th crease for odd $i$, considering the edges of the paper the $0$th and $(n+1)$st crease.
Let the \emph{complement} of the symbols in the string $S$ be defined as $\text{comp}(S[i])=S[i]$ (resp., $\text{comp}(S[i])=-S[i]$) if $i$ is odd (resp., even). This definition is motivated by the observation that, if we consider what happens to the section of paper to the right of crease $i$ when $i$ is folded, we see that the string $S[i+1 \dots 2n+1]$ is converted to $\text{comp}(S[2n+1 \dots i+1])$.  
We show an algorithm that finds the leftmost legal fold defined by the $k$th crease, if one exists, in time $O(k)$.
We then show that successive applications of this algorithm will find a solution for the problem in $O(n)$ time.

The following algorithm finds the smallest prefix of length $2k+1\le 2n+1$ of the input $S$ such that: 
\begin{enumerate*}[label=(\roman*)]
\item \label{cond:boundary} $S[1]\le S[2k-1]$; 
\item \label{cond:assign} for $i\in[1,k-2]$, then $S[k-i]=\text{comp}(S[k+i])$.
\end{enumerate*} 
Importantly, if conditions~\ref{cond:boundary} and~\ref{cond:assign} are satisfied, then a fold through crease $k$ is legal. Note, however, that not every legal fold is necessarily such a satisfying prefix; some may be a satisfying prefix of the \emph{reversed} string $S^R$ (or equivalently, a \emph{suffix} of $S$). Fortunately, it is easy to handle both cases simultaneously as we will demonstrate in Theorem~\ref{thm:1d} below.

This algorithm is a minor variation of the algorithm in~\cite{Palindrome} that we show here for completeness. 
We say that a position $f>k$ fails a crease $S[k]$ if $S[k-(f-k)]\neq\text{comp}(S[f])$, i.e., $(f-k)$ violates (ii) for crease $S[k]$.
The array $F$ stores the \emph{failure number} of creases. 
The failure number $F[k]$ of a crease $S[k]$ is defined as the minimum $f>k$ that fails  $S[k]$.

\begin{algorithm}[H]\DontPrintSemicolon
	\LinesNumbered
	\SetKw{KwGoTo}{go to}
\textbf{Algorithm~$A$}\;
\KwData{string $S$ of length $2n+1$.}
		$k\leftarrow 2$\;
		\For{$it\leftarrow 3$ \KwTo $2n+1$}{
			$i\leftarrow it-k$\;\label{start}
			\If{$i<1$}{\Continue}
			\tcc{Check condition (i)}
			\If{$k-i=1$  {\bf and} $S[1]\le(S[2k+1])$}{\KwRet{$k$}}
			\tcc{Check condition (ii)}
			\If{$S[k-i]\neq \text{comp}(S[k+i])$}{
				$F[k]\leftarrow i$\;\label{failure}
				$k'\leftarrow k+2$\;
			\While {$k'<it$}{
				\If{$it-k'\neq F[k-(k'-k)]$}{$F[k']\leftarrow\min\{it-k',F[k-(k'-k)]\}$}
				\Else{
					$k\leftarrow k'$\;
				\KwGoTo \ref{start}\;}
			
			$k'\leftarrow k'+2$\;
	}
$k\leftarrow k'$\;
		}
		}
		\Return $0$.
\end{algorithm}

\begin{figure}[h!]
	\centering
	\includegraphics[width=\columnwidth]{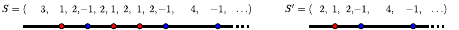}
	\caption{An instance $S$ and the result $S'$ obtained by folding through $S[6]$. A fold though $S[2]$ violates~\ref{cond:boundary}, and attempting to fold through $S[4]$ is failed by crease $S[6]$ and, hence, violates~\ref{cond:assign}.}
	\label{fig:1d}
\end{figure}

\begin{lemma}
	\label{lem:1D-fold}
	Algorithm $A$ either finds the leftmost legal all-layers simple fold at crease $S[k], 1<k\le n+1$ in $O(k)$ time  or returns $0$ if no such fold exists in $O(n)$ time.
\end{lemma}
\begin{proof}
	The Lemma follows from~\cite{Palindrome} and~\cite{arkin04}.
	For completeness, we give here a summary of the arguments.
	Conditions~\ref{cond:boundary} and~\ref{cond:assign} for odd (resp., even) indices $i$ imply that a fold at crease $S[k]$ satisfies condition~\ref{cond:creases-to-creases} (resp.,~\ref{cond:mountains-to-mountains}) from Definition~\ref{def:legal}.
	Hence, if a value $k\neq 0$ is returned, it defines a legal fold.
	Similarly, every legal fold at $S[k]$ satisfies conditions~\ref{cond:boundary} and~\ref{cond:assign}.
	The integer variables $k$ and $it$ only increase and are upper bounded by $2n+1$.
	When a value of $k$ is returned, $it=O(k)$ and, thus, the algorithm runs in $O(k)$ (resp., $O(n)$) time for positive (resp., negative) instances.
	The algorithm maintains the invariant that every crease to the left of $k$ have been assigned a failure number, i.e., there exists no legal fold to the left of $k$.
	
	The algorithm reaches line~\ref{failure} if position $it$ fails $S[k]$.
	For all creases $S[k']$ where $k<k'<it$, we can determine in constant time if a position to the left of $it$ fails $S[k']$ as follows.
	Let $S[k'']$ be the crease symmetric to $S[k']$ about $S[k]$, i.e., $k''=k-(k'-k)$.
	If $F[k'']<it-k'$, then by~\ref{cond:assign} with respect to $S[k]$ we have $S[k''+F[k'']]=\text{comp}(S[k'-F[k'']])$ and $S[k''-F[k'']]=\text{comp}(S[k'+F[k'']])$.
	By the definition of failure number, $F[k'']$ is the smallest integer such that $$\text{comp}(S[k'-F[k'']])\neq\text{comp}(\text{comp}(S[k'+F[k'']]))=S[k'+F[k'']].$$
	Notice that $a=\text{comp}(\text{comp}(a))$ and $a\neq \text{comp}(b)\rightarrow \text{comp}(a)\neq b$ for all $a$ and $b$.
	Hence, $F[k']=F[k'']$ and a fold through $S[k']$ is illegal.
	If $F[k'']>it-k'$, then $S[k''+(it-k')]=\text{comp}(S[k''-(it-k')])=S[k'-(it-k')]$.
	Because the position $it$ fails $S[k]$, $$S[k'+(it-k')]\neq \text{comp}(S[k''-(it-k')])=\text{comp}(S[k'-(it-k')]).$$
	Therefore, $F[k']=it-k'$ and a fold through $S[k']$ is illegal.
	Then, the algorithm finds the leftmost crease $k$ for which no position between $k$ and $it$ fails $S[k]$.
	Therefore, the algorithm returns the leftmost legal crease on the left half of $S$.
\end{proof}

\begin{theorem}
	\label{thm:1d}
	All-layers simple foldability of a 1D assigned or unassigned crease patterns can be decided in deterministic linear time.
\end{theorem}

\begin{proof}
	We provide a constructive proof.
	Run Algorithm $A$ on inputs $S$ and $S^R$, a reversed copy of $S$, in parallel.
	If both return $0$, then by Lemma~\ref{lem:1D-fold}, $S$ is not all-layers simple foldable.
	Else, let $k$ be the first value returned.
	Without loss of generality, $k$ was returned by Algorithm $A$ on $S$.
	Using our definitions for folding at $S[k]$, the resulting crease pattern $S'$ is represented by the substring of $S$ from $k+1$ to $2n+1$.
	Hence we generated a smaller subproblem of size $2n-k$ in $O(k)$ time.
	By Lemma 4.1 of~\cite{arkin04}, if any legal fold can be done in $S$, the resulting crease pattern $S'$ is all-layers simple foldable if and only if $S$ also is.
	Then, by successively applying the above algorithm in the resulting smaller subproblems we arrive at a crease pattern $(P_m,\emptyset)$ in $O(n)$ time if and only if $S$ is all-layers simple foldable.
\end{proof}


\section{Unassigned Orthogonal 2D Crease Patterns}
\label{sec:Unassigned}

Let $(P,\Sigma)$ be the input for our problem, where $P$ is an orthogonal polygon and $\Sigma$ contains axis-aligned creases.
We first show a necessary condition that any crease pattern of a solution $((P_1,\Sigma_1), \ldots, (P_k,\emptyset))$ must satisfy.

\begin{lemma}
\label{lem:inf-extention}
Given a crease pattern $(P,\Sigma)$ with orthogonal paper $P$ and orthogonal creases $\Sigma$, if a line $\ell^*$ in $\mathbb{P}$ contains a crease in $\Sigma$, then every point in $P\cap \ell^*$ must be on a crease in $\Sigma$ or else $(P,\Sigma)$ is not infinite all-layers simple foldable.
\end{lemma}
\begin{proof}
Recall that, by condition~\ref{cond:only-on-crease}, a simple fold through a line $\ell$ does not contain any point of a facet.
By contradiction, let $\ell^*$ be a line containing a crease such that $P\cap \ell^*\setminus\partial P$ is not covered by creases in $\Sigma$.
We prove by induction that $(P,\Sigma)$ is not simple foldable.
Trivially, $\ell^*$ cannot be an axis of a simple fold by~\ref{cond:only-on-crease}.
The base case is when all creases in $\Sigma\neq\emptyset$ are on an axis that define illegal simple folds, hence, no legal fold is possible.
Else, a legal fold through some axis $\ell$ is possible.
By~\ref{cond:creases-to-creases}, a simple fold maps points on creases to creases and facet points to facet points.
If $\ell^*$ is parallel to $\ell$, the creases and facet points on $\ell^*$ will be mapped to a creases and interior points on a line $\ell'$ ($\ell'=\ell^*$ if $\ell$ is to the left of $\ell$ or $\ell'$ is the reflection of $\ell^*$ about $\ell$).
Else, $\ell^*$ is perpendicular to $\ell$ and $\ell^*$ will also contain a crease and a facet point of $(P',\Sigma')$ after the fold.
Because $|\Sigma'|<|\Sigma|$, we are done.
%
\end{proof}

We now describe the encoding of the input $(P,\Sigma)$.
Let $n_P$ be the number of vertices in $P$.
Let $L=\{\ell_1,\ldots,\ell_{n_L}\}$ be the set of lines in $\mathbb{P}$ that contain creases in $\Sigma$.
By Lemma~\ref{lem:inf-extention}, it is enough to store $L$ because $\Sigma$ can be inferred by $L$ and $P$ or else $(P,\Sigma)$ is a negative instance.
Our algorithm uses $(P,L)$ as the input and define its size as $n_P+n_L$.
First, we provide an algorithm for when $P$ is a rectangle.

\begin{lemma}
\label{lem:UnassignRect}
Let $(P,L)$ be an unassigned crease pattern where $P$ is a rectangle and lines in $L$ are axis-aligned.
Infinite all-layers simple foldability of $(P,L)$ can be determined in linear time.
\end{lemma}
\begin{proof}
We reduce the problem to two 1D instances, each solvable in $O(n')$ time where $n'$ is the number of creases in the 1D instance by Lemma~\ref{thm:1d}.
Let $L_h$ (resp., $L_v$) be the subset of $L$ containing horizontal (resp., vertical) lines.
Build a 1D instance with paper $P_v$ (resp., $P_h$) being a line segment congruent to the left (resp., bottom) edge of $P$.
Add a crease in $P_v$ (resp., $P_h$) at the intersections between the lines in $L_h$ (resp., $L_v$) and its correspondent edge.
By definition, any legal fold though a line in $L_h$ (resp., $L_v$) is also a legal fold in its corresponding crease in $P_v$ (resp., $P_h$), and its resulting crease pattern $(P',L')$ can be converted into 1D problems $P_v'$ and $P_h'$ such that $P_v'$ (resp., $P_h'$) is equal to the resulting crease patterns of the fold in $P_v$ (resp., $P_h$), and $P_h=P_h'$ (resp., $P_v=P_v'$).
Therefore, any folding sequence of the instance $(P,L)$ satisfying the problem can be converted to two folding sequences of the 1D problems and vice versa.
\end{proof}

Theorem 8 of \cite{SimpleFolds_JCDCGGG2016}, which says that instances of simple foldability with rectangular paper and orthogonal creases are equivalent in the all-layers model and infinite all-layers model, implies that the reduction in Lemma~\ref{lem:UnassignRect} also applies to the all-layers model.
Now, we prove our algorithmic result.




\begin{figure}[h!]
	\centering
	\includegraphics[width=.5\columnwidth]{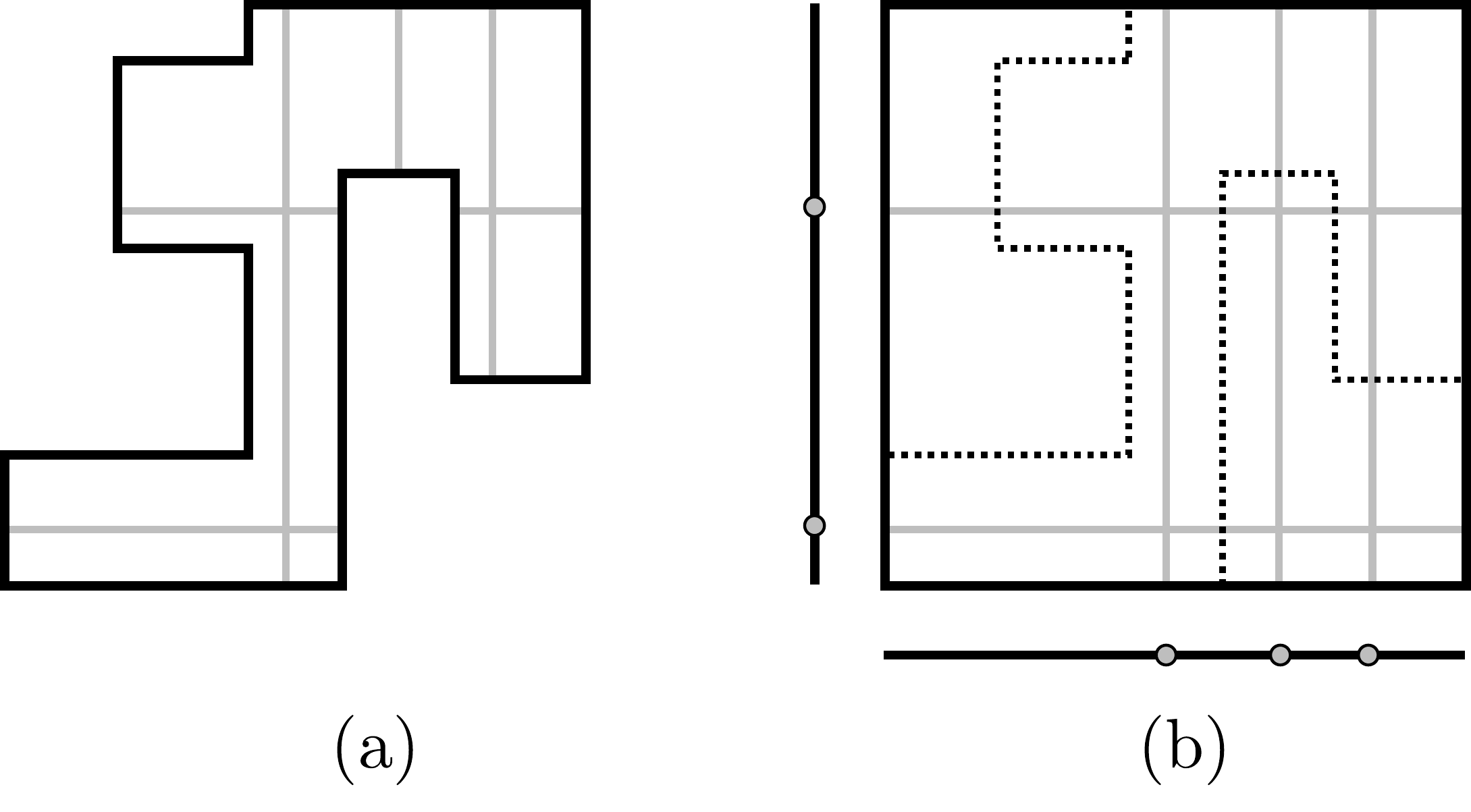}
	\caption{(a) An instance $(P,L)$ where $P$ is an orthogonal polygon and (b) its corresponding instance $(R,L)$ where $R$ is the bounding box of $P$ and the boundary of $P$ is shown with dotted line segments. The 1D instances corresponding to $(R,L)$ are also shown in (b).}
	\label{fig:2d}
\end{figure}

\begin{theorem} \label{thm:OrthOrth}
Determining whether an unassigned crease pattern on orthogonal paper with orthogonal creases has an infinite all-layers simple folding can be solved in linear time.
\end{theorem}
\begin{proof}
We reduce to the rectangular paper/orthogonal creases problem solved in Lemma~\ref{lem:UnassignRect}.
Let the instance with orthogonal paper and creases be represented by $(P,L)$ of size $n=n_1+n_2$ where $n_1$ and $n_2$ are the sizes of $P$ and $L$, respectively.
We can obtain the bounding box $R$ of $P$ in $O(n_1)$ time.
Output the instance $(R,L)$ that can be solved in $O(n_2)$ time by Lemma~\ref{lem:UnassignRect}.
We now show that $(R,L)$ admits an infinite all-layers simple folds sequence if and only if $(P,L)$ also does.

The backward implication is straightforward. 
If $(R,L)$ has a simple folding, then so must $(P,L)$. 
Notice that creases induced by $L$ in $P$ are a subset of the creases induced by $L$ in $R$.
Because $P$ is a subset of $R$, by definition, any fold made on $R$ will also be possible on $P$, so $P$ is infinite all-layers simple foldable.

Now we show that, if $((P_1,L_1),\ldots,(P_k,L_k))$ is an infinite all-layers simple folding for $(P,L)$, then $((R_1,L_1),\ldots,(R_k,L_k))$
is an infinite all-layers simple folding for $(R,L)$, where $R_i$ is the bounding box of $P_i$ for $i\in\{1,2,\ldots,k\}$.
We show that the axis of the simple fold from $(P_i,L_i)$ to $(P_{i+1},L_{i+1})$ is also the axis of a simple fold that transforms $(R_i,L_i)$ into $(R_{i+1},L_{i+1})$.
By definition, $P_i$ and $R_i$ share at least one each of bottom, left, top, and rightmost points, which we denote by $x_b$, $x_l$, $x_t$ and $x_r$ respectively.
Without loss of generality, take a vertical axis $\ell$ pointing up  whose supporting line is in $L_i$ that defines a legal fold.
By definition, the simple fold through $\ell$ reflects the right  portion of $P_i$ and $R_i$ to the left of $\ell$ producing $P_{i+1}$ and $R_i'$. 
Lemma~\ref{lem:inf-extention} guarantees that both originate the same set of new supporting lines $L_{i+1}$.
Additionally, in both cases, the rightmost point of the paper after the fold will be on $\ell$.
If the reflection of $x_r$ becomes the leftmost point in $P_{i+1}$ it will also become the leftmost point of $R_i'$.
Because $x_t$ and $x_b$ don't change y-coordinates, $R_i'=R_{i+1}$.
\end{proof}

\section{Partially Assigned Orthogonal Crease Patterns}
\label{sec:PartAssign}
We prove that, given a crease pattern with a subset of creases assigned mountain/valley, it is NP-complete to decide infinite all-layers simple foldability, even if the crease pattern is a square grid and the paper is rectangular.

\begin{lemma}
\label{lm:Hardness}
It is strongly NP-hard to decide infinite all-layers simple foldability of partially assigned crease patterns, even if the paper is rectangular and the creases form a square grid. 
\end{lemma}

\begin{figure}[h!]
	\centering
	\includegraphics[width=.79\columnwidth]{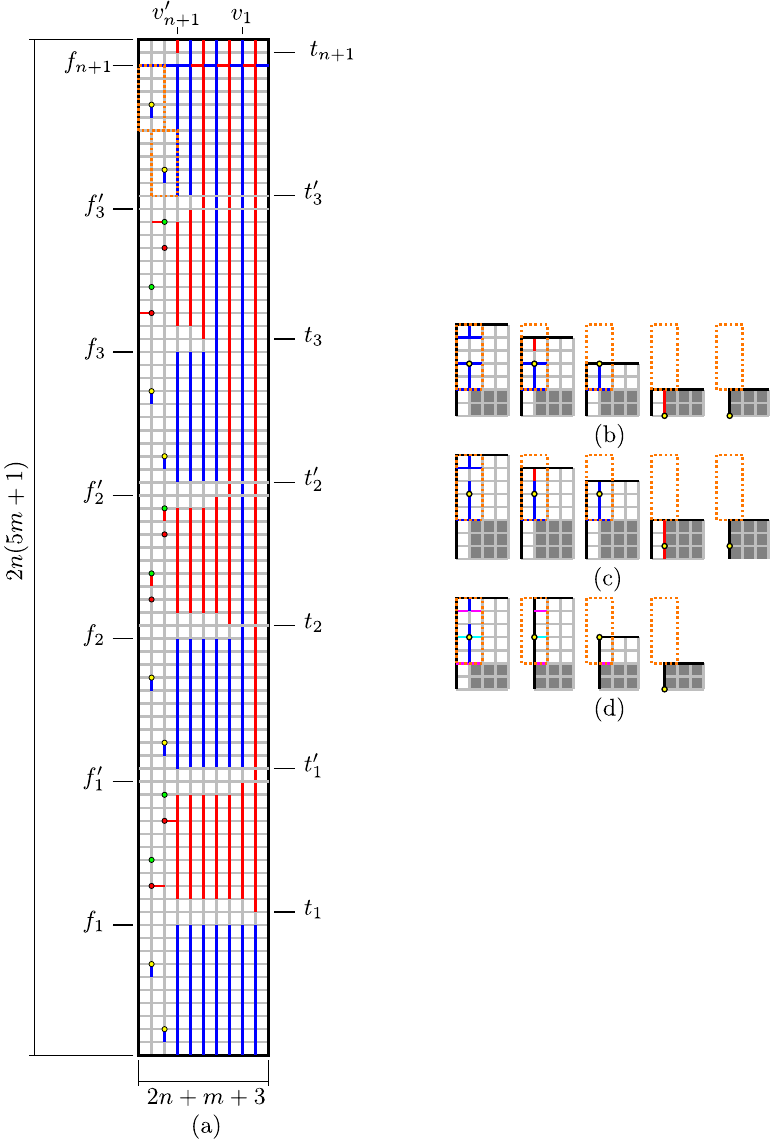}
	\caption{(a) Reduction from 3SAT to partially assigned simple foldability under the infinite all-layers model. The instance shown corresponds to the boolean formula $(x_1\vee \overline{x_2} \vee x_3)\wedge(x_1\vee \overline{x_2} \vee \overline{x_3})$.
		A satisfied clause gadget can be folded using one of the folding sequences in (b)--(d).}
	\label{fig:reduction}
\end{figure}

\begin{proof}
We provide a reduction from 3SAT.
An instance of 3SAT is given by a set of $n$ variables $\{x_1,\ldots,x_n\}$ and a boolean formula in conjunctive normal form with $m$ clauses of the form $(l_1\vee l_2\vee l_3)$ where $l_k$, $k\in\{1,2,3\}$, is the boolean value of a variable $x_i$, $i\in\{1,\ldots,n\}$, or its negated value denoted by $\overline{x_i}$.
The 3SAT problem asks whether there exists an assignment from the variables to $\{\textsc{true},\textsc{false}\}$ such that the boolean formula evaluates to \textsc{true}.
Given the 3SAT instance described above, we construct a crease pattern on a $2n(5m+1)\times (2n+m+3)$ rectangular paper as shown in Figure~\ref{fig:reduction} (a).
Informally, a variable $x_i$ is represented by the section of paper between rows $f_i$ and $t_i'$ inclusive. The yellow dots represent literals in a clause and folding a yellow dot onto a green or red dot represents assignment as positive or negative respectively. Yellow dots (literals) sharing a vertical line are in the same clause.

An overview of the proof follows: because of~\ref{cond:consistent-assign}, the partial assignment forces any sequence of legal simple folds to choose between folding through $t_1$ or $f_1$, bringing them on top of a green or red dot, respectively, which encodes the boolean assignment of the variable $x_1$.
After a vertical fold on the right edge of the paper, the construction forces these yellow dots to coincide with the yellow dot right above it by folding through $t_1'$ or $f_1'$, which adds a valley assignment to a crease next to a yellow dot of a corresponding clause if its literal of $x_1$ evaluates to \textsc{false}.
We apply induction on the resulting crease pattern and, after choosing the assignment of all variables, the topmost set of yellow points will have at least one crease that is not assigned valley if and only if the SAT instance has a positive solution.
If a vertex has all four incident valley-assigned creases, there is no legal simple fold that folds through any of its creases.

Now, the complete proof: consider the bottom left corner as the origin.
In the following, we use the indices $i\in\{1,\ldots,n\}$ to refer to variables and $j\in\{1,\ldots,m\}$ to refer to clauses.
We first define the position of some points that will serve as reference to the construction.
For each clause $c_j$, define the points $p_{j,i} = (j+1, 10 m i -5 j - 5 m +  2 i)$ and a point $p_{j,(n+1)}$ using the same formula, that represent $n+1$ copies of the clause. 
In Figure~\ref{fig:reduction} (a), such points are represented as yellow dots.
Assign the crease right below each $p_{j,i}$ as valley.
Also define $r_{j,i}=(j+1,5 j - 10 m + i (10 m + 2) + 6)$ and $g_{j,i}=r_{j,i}+(0,2)$, represented in Figure~\ref{fig:reduction} (a) as red and green dots respectively.
If the first (resp., second, resp., third) literal that $c_j$ contain is $x_i$, assign the crease to the right (resp., bottom, resp., left) of $r_{j,i}$ as mountain.
If the literal is negated, i.e., $\overline{x_i}$, assign the crease to the right (resp., bottom, resp., left) of $g_{j,k}$ as mountain.

We now define some reference horizontal lines.
Let $f_i$ be the horizontal line appearing at height ${5m+2(i-1)(5m+1)}$, and $t_i, f_i'$ and $t_i'$ the horizontal lines $1$, $5m+1$ and $5m+2$ units above  $f_i$ respectively.
Define $f_{n+1}$ and $t_{n+1}$ analogously. 
For each variable $x_i$, we define the crease assignment of two vertical lines $v_i$ and $v_i'$, which are $m + 2 n - 2 i + 2$ and $m + 2 n - 2 i + 3$ units to the right of the origin respectively.
Also define $v_{n+1}'$, which is $m+1$ units to the right of the origin.
Assign $5m$ creases bellow $f_k$ as valley for $k\le i$.
Assign $5m-2$ creases bellow $t_k'$ as mountain for $k< i$, skipping the first one.
Assign mountain to $5m-1$ creases below $f_i'$ in $v_i$ and to all the creases above $t_i$ in $v_i'$.
Assign valley to all the creases above $t_i'$ in $v_i$.
Finally, we describe the assignment of creases in $f_{n+1}$.
Assign the $m+2$ leftmost creases as valley and alternate between mountain and valley until the end of the line.
All other creases in the construction that do not have any assignment remain unassigned.

We now show that, if the 3SAT instance has a positive solution, so does the constructed instance of simple foldability.
In order from 1 to $n$, if $x_i$ is assigned \textsc{true} (resp., \textsc{false}), valley fold through $t_i$ (resp., $f_i$), bringing the bottom part up and aligning all $p_{j,i}$ with $g_{j,i}$ (resp., $r_{j,i}$);
mountain fold though $v_i'$, bringing the right strip of paper to the left;
valley fold through $t_i'$ (resp., $f_i'$), aligning all $p_{j,i}$ with $p_{j,i+1}$;
valley fold through $v_i$.
If $x_i$ appears as the first (resp., second, resp., third) literal of $c_j$ which  evaluates to \textsc{false}, then these folding sequence brings a valley-assigned crease to the right (resp., top, resp., left) of $p_{j,i+1}$.
After all $n$ steps, valley fold through $f_{n+1}$, bringing the top strip of paper down, and then through $v_{n+1'}$.
This eliminates all assigned creases apart from the ones adjacent to reference points.
Notice that, for a clause $c_j$, the corresponding assigned creases (which are adjacent to a reference point with index $j$) are contained in a $2\times 5$ bounding box containing $p_{j,n+1}$ in its $(1,2)$ relative coordinate (shown as a dotted orange box in Figure~\ref{fig:reduction}).
All boxes containing $p_{k,n+1}, k<j$ are below and to the right of the point $(0,-2)$ relative to $p_{j,n+1}$ (darkened region shown in Figure~\ref{fig:reduction} (b)--(d)).
We start  by assuming that all folds relative to $c_k$, $k<j$ are already folded and the top left corner of the paper coincides with a corner of the box containing $p_{j,n+1}$.
If the second literal of $c_j$ evaluates to \textsc{true}, the assigned creases in the bounding box of $p_{j,n+1}$ are a subset of the leftmost crease pattern in Figure~\ref{fig:reduction}(b). If both the first and third literal evaluate to \textsc{true}, then the assigned creases in the bounding box of $p_{j,n+1}$ are a subset of the pattern in Figure~\ref{fig:reduction}(c).
We can fold through all creases in such box using the folding sequence in Figure~\ref{fig:reduction}(b) (resp., (c)), which does not introduce any assignment in boxes of $p_{l,n+1}$, $l>j$.
Else, the assigned creases are a subset of the leftmost crease pattern in Figure~\ref{fig:reduction}(d), where exactly one of the \textcolor{cyan}{cyan} creases and one of the \textcolor{magenta}{magenta} creases are assigned valley.
By following the folding sequence in Figure~\ref{fig:reduction}(d), we eliminate all assigned creases in the box of $p_{j,n+1}$ while not introducing any assignment in boxes of $p_{l,n+1}$, $l>j$.
By induction, we can continue folding until there are no assigned creases and the resulting crease pattern has a positive solution using the algorithm for unassigned crease patterns.


Now, we show that if the produced crease pattern is foldable, the 3SAT instance has a positive solution.
By~\ref{cond:consistent-assign}, no vertex can have 4 adjacent creases with the same assignment or else there is no legal fold through such vertex.
Initially, the only possible legal folds are through $f_1$ or $t_1$ due to (4).
First, assume that the solution folds through $f_1$, bringing $p_{j,1}$ onto $r_{j,1}$.
That assigns some creases between $f_1$ and $t_1$ making it impossible to fold through $t_1$. 
The only possible simple fold now is through $v_1'$.
Then, folding through $t_1'$ is illegal since it would bring a mountain-assigned crease in $v_2'$ right above $f_1$ onto a mountain-assigned crease right above $t_2$.
Hence, the only possible fold is through $f_1'$, followed by a valley fold through $v_1$.
The resulting folded state has a height-1 strip of paper in the bottom that does not contain any assigned crease. 
If this is folded on top of any other part of the construction, it will not create any other assignment. 
At this state, the solution can either fold through the horizontal line right above $t_1'$ and then fold through $f_2$ or $t_2$, or fold directly through $f_2$ or $t_2$, which are the only available legal folds.
In both scenarios, the resulting assigned pattern and the boundary of the paper are the same.
Therefore, after folding through $v_1$, we obtain a smaller version of the reduction with variables $\{x_2,\ldots,x_n\}$ and some extra assigned creases in the vicinity of $p_{j,2}$ if $c_j$ contained a literal of $x_1$.
In particular, if said literal is positive, there exist one extra crease adjacent to  $p_{j,2}$ that is assigned valley.
Now, assume that the solution folds through $t_1$. 
The next folds must be through $v_1'$ followed by $t_1'$, and $v_1$ or else an assigned crease will be mapped onto a crease of same assignment contradicting~\ref{cond:consistent-assign}.
We again obtain a smaller version of the reduction and if $c_j$ contained the literal $\overline{x_1}$, there is one extra crease adjacent to  $p_{j,2}$ that is assigned valley.
By induction, the solution must fold in this manner (bringing $p_{j,i}$ onto $p_{j,i+1}$) until all of the yellow marked points are on top of a $p_{j,n+1}$.
Because the crease pattern in foldable, there must exist at least one crease adjacent to $p_{j,n+1}$ that is not assigned valley.
That corresponds to an assignment in which at least one literal in each clause evaluates to \textsc{true}.
\end{proof}

\begin{theorem}
It is strongly NP-complete to decide simple foldability under the infinite all-layers model of a partially assigned crease pattern, even if the paper is rectangular and the creases form a square grid. 
\end{theorem}
\begin{proof}
Because unfolding is forbidden, a simple fold reduces the number of creases by at least one.
Therefore, a sequence of simple folds that folds through all creases of a crease pattern with $n$ creases can have $O(n)$ folds.
We can check in $O(n)$ time whether a given axis define a legal simple fold.
The rest of the proof follows from Lemma~\ref{lm:Hardness}.
\end{proof}

\section{Open Problems}
\label{sec:Open}

Lemma 4.1 of~\cite{arkin04}, used in Theorem~\ref{thm:1d}, only applies to assigned or unassigned 1D crease patterns.
Hence, we leave open the algorithmic complexity of deciding all-layers simple foldability of partially assigned 1D crease patterns.
Indeed, for this class of input, the problem has not been studied in the other simple foldability models introduced in~\cite{arkin04}.
We conjecture that the algorithm presented in Section~\ref{sec:1D} also finds a solution if one exists.
The main remaining open problem is whether infinite all-layers simple
foldability can be solved in polynomial time, or is NP-hard, in fully assigned
crease patterns~\cite{SimpleFolds_JCDCGGG2016}.
This problem remains open in particular when restricted to orthogonal crease patterns.
We also leave open the complexity of infinite all-layers simple foldability in unassigned nonorthogonal crease patterns, for example, axis-aligned rectangular pieces of paper with unassigned creases at multiples of $45^\circ$, or even general unassigned crease patterns on general polygonal pieces of paper.

\section*{Acknowledgments}

This work was initiated during an open problem session associated with MIT
class 6.849 (Geometric Folding Algorithms) taught by E. Demaine and J. Ku in
Spring 2017.  We thank the other participants of that class for providing a
stimulating research environment.

Supported in part by NSF grants EFRI-1240383, CCF-1138967 and CCF-1422311, and
the Science without Borders scholarship program.


\small
\bibliographystyle{alpha}
\bibliography{main}

\end{document}